\documentclass[conference,compsoc]{IEEEtran}
\IEEEoverridecommandlockouts

\usepackage{cite}
\usepackage{amsmath,amssymb,amsfonts,amsthm}
\usepackage{algorithm}
\usepackage[noend]{algpseudocode}
\usepackage{graphicx}
\usepackage{xcolor}
\usepackage{booktabs}
\usepackage{url}
\usepackage{tikz}
\usetikzlibrary{arrows.meta,positioning,fit,backgrounds,calc,
  shapes.geometric,shapes.multipart,decorations.pathreplacing}
\usepackage[caption=false,font=footnotesize]{subfig}
\usepackage{balance}

\usepackage{float}

\definecolor{gibNavy}{RGB}{27,38,59}
\definecolor{gibBlue}{RGB}{52,101,164}
\definecolor{gibTeal}{RGB}{23,147,140}
\definecolor{gibRed}{RGB}{192,57,43}
\definecolor{gibAmber}{RGB}{214,137,16}
\definecolor{gibGreen}{RGB}{39,152,96}
\definecolor{gibGray}{RGB}{236,238,241}

\newcommand{\sys}{\textsc{Garlic-in-a-Box}}
\newcommand{\code}[1]{\texttt{\small #1}}
\newcommand{\Fv}{\mathcal{F}_v}
\newtheorem{definition}{Definition}
\newtheorem{theorem}{Theorem}

\begin{document}
\title{Garlic in a Box: I2P Network Emulator}

\author{
\IEEEauthorblockN{\small\bfseries
\mbox{Siddique Abubakr Muntaka\IEEEauthorrefmark{1}},
\mbox{Faiza Hussein\IEEEauthorrefmark{2}},
\mbox{Boutros Farah\IEEEauthorrefmark{3}},
\mbox{Anthony Kashi\IEEEauthorrefmark{3}},
\mbox{Jacques Bou Abdo\IEEEauthorrefmark{1}},
\mbox{Jacques Demerjian\IEEEauthorrefmark{3}\IEEEauthorrefmark{4}}
}
\IEEEauthorblockA{
\mbox{\IEEEauthorrefmark{1}School of IT, University of Cincinnati, OH, USA}\quad
\mbox{\IEEEauthorrefmark{2}Dept. of Computer Science, Garden City University College, Ghana}\\
\mbox{\IEEEauthorrefmark{3}LaRRIS, Faculty of Sciences, Lebanese University, Lebanon}\quad
\mbox{\IEEEauthorrefmark{4}Faculty of Arts and Sciences, USEK, Lebanon}\\
\{muntaksr@mail.uc.edu, bouabdjs@ucmail.uc.edu\}
}
}


\maketitle

\begin{abstract}
The Invisible Internet Project (I2P) is a peer-to-peer anonymity network
whose routing and peer-selection mechanics underlie a growing family of
state-sponsored covert relay architectures. Studying these mechanics
demands a controlled, reproducible environment. We present
\sys{}, an installable emulator that deploys a fully isolated network of
unmodified I2P Java routers on a single host or cluster of hosts. Each router runs in its own
Linux network namespace. A filesystem bootstrap removes the dependency on
signed reseed servers, and synthetic geographic diversity activates
the same location-aware peer-selection code paths present on the live
network. A fifty-router deployment across five geographic locations stabilises
in about twenty-five minutes and uses roughly 21\,GB of memory.
To demonstrate the platform's value, we model I2P peer selection as a
two-stage process and conduct a reputation attack. Across five independent
trials, the attack promotes the adversary into the victim's eligible selection set with perfect consistency, and the adversary is chosen for
$1.38\times$ more tunnels than matched controls. The platform and its
measurement pipeline are released to support reproducible I2P security
research and hands-on cyber-defence training.
\end{abstract}


\section{Introduction}
\label{sec:intro}
The Invisible Internet Project (I2P) is a peer-to-peer anonymity
network in which no router is merely a client \cite{muntaka2026systemic}. Each one relays for
others, and each one chooses the peers that make up its own tunnels
from a reputation database it maintains locally. Those choices decide
who can observe a flow, who can correlate it, and how much anonymity
the network ultimately provides. They are also the point where the
network is most open to manipulation, since an adversary able to
influence those choices can raise its own standing in a victim's
selection process.

This has become more than a theoretical concern. Operational Relay Box
(ORB) networks operated by actors such as APT28, Volt Typhoon, and
SPACEHOP have begun to reuse the peer-selection and relay-management
techniques of anonymity systems, I2P prominent among them, spreading across compromised infrastructure on a scale no volunteer network
approaches~\cite{mandiant_orb2024,mitre_apt28,mitre_volttyphoon,mitre_spacehop}. A Cybersecurity and Infrastructure Security Agency (CISA) advisory in April 2026 flagged these
covert ecosystems as a strategic national security threat \cite{cisa2026}. Whether the
goal is to disrupt such traffic or to train the operators who must
respond to it, the first requirement is the same: a clear understanding
of how these networks decide where traffic goes.

That understanding is difficult to obtain. The live I2P network cannot
be experimented on without raising ethical and reproducibility
concerns, and a handful of manually assembled routers cannot capture
the network-level behavior that a realistic deployment produces.
Neither option gives an educator a safe place to teach these attacks in
the first place.

We present \sys{}, an installable emulator that closes this gap.
It runs unmodified I2P Java routers, each isolated in a Linux network namespace, and provides a control interface, a scenario engine, and a measurement pipeline. The platform answers two needs at once. It gives researchers a controlled instrument for observing router behavior and collecting protocol data at scale, from peer selection and tunnel construction to any interaction the unmodified router exposes. It also gives defenders a safe range on which to
rehearse against the protocol mechanics that covert relay networks now
exploit.

\smallskip\noindent In support of these goals, this paper contributes the following.
\begin{enumerate}
  \item A packaged, reproducible I2P emulator
  (Section~\ref{sec:design}) that runs unmodified routers in isolated
  namespaces and bootstraps without a signed reseed server.
  
  \item A two-stage probabilistic model of I2P peer selection (Section~\ref{sec:model}), distinguishing the selection stage, where a peer becomes eligible, from the delivery stage, where it must actually relay traffic.

  \item A profiles-as-measurement method
  (Section~\ref{sec:method}) that reads the victim's own on-disk peer state to quantify attack effects on an unmodified router.
  \item A five-trial empirical study (Section~\ref{sec:eval})
  of a reputation attack, with a deterministic result on the
  selection stage and a precise account of the delivery stage.
  \item A structured training capability (Section~\ref{sec:training}) that supports red-team and blue-team exercises on anonymous overlay threats, including scenarios relevant to ORB-style covert relay operations.
\end{enumerate}

\section{Background and Related Work}
\label{sec:bg}

I2P is a message-oriented anonymous overlay \cite{muntaka2025optimizing}. Routers publish
\emph{RouterInfo} records to a distributed hash table called the
\emph{netDb}, which is replicated by a subset of nodes designated as
\emph{floodfill} peers~\cite{i2p-techdoc}. Traffic travels through
paired unidirectional \emph{tunnels}, rebuilt approximately every ten
minutes \cite{abdo2023modeling}. A router selects the peers in its own tunnels from a local
reputation database \cite{muntaka2026bridging}. The security properties of the network follow
directly from those selection decisions.

Each I2P router keeps a \emph{profile} for every peer it has interacted
with, accumulating observed bandwidth, tunnel-build acceptance and
rejection rates, and response latency~\cite{i2p-peersel}. The
\emph{ProfileOrganizer} sorts these profiles into tiers on a coalescing
timer, admitting the best-observed peers into a \emph{fast} tier and a
\emph{high-capacity} tier from which tunnel hops are drawn~\cite{schimmer09}.
Alongside these measured quantities, every profile keeps two explicit bonus
fields, \code{speedBonus} and \code{capacityBonus}. Their effect is to nudge a
peer's standing without reference to anything it has actually done, which is
precisely what makes them attractive to an attacker: manipulating profile state
amounts to writing directly to these fields rather than earning a reputation
through behaviour. Whether such a nudge is enough to matter was settled by
Herrmann and Grothoff, whose real-world study of I2P found that performance-based
selection already introduces measurable bias even when peers behave
benignly~\cite{herrmann11}. If honest variation in observed performance is enough
to skew selection, then a small change to a peer's stored state can just as
easily push it from one tier into another.

Fig.~\ref{fig:i2p-arch} follows this process from start to finish: a new router
first gathers peers through reseed, maintains and updates their profiles in a
local database, and then, when assembling tunnels, draws its hops preferentially
from the fast and high-capacity groups that the ProfileOrganizer curates \cite{muntaka2026geographic}.

\begin{figure}[H]
\centering
\includegraphics[width=0.90\linewidth]{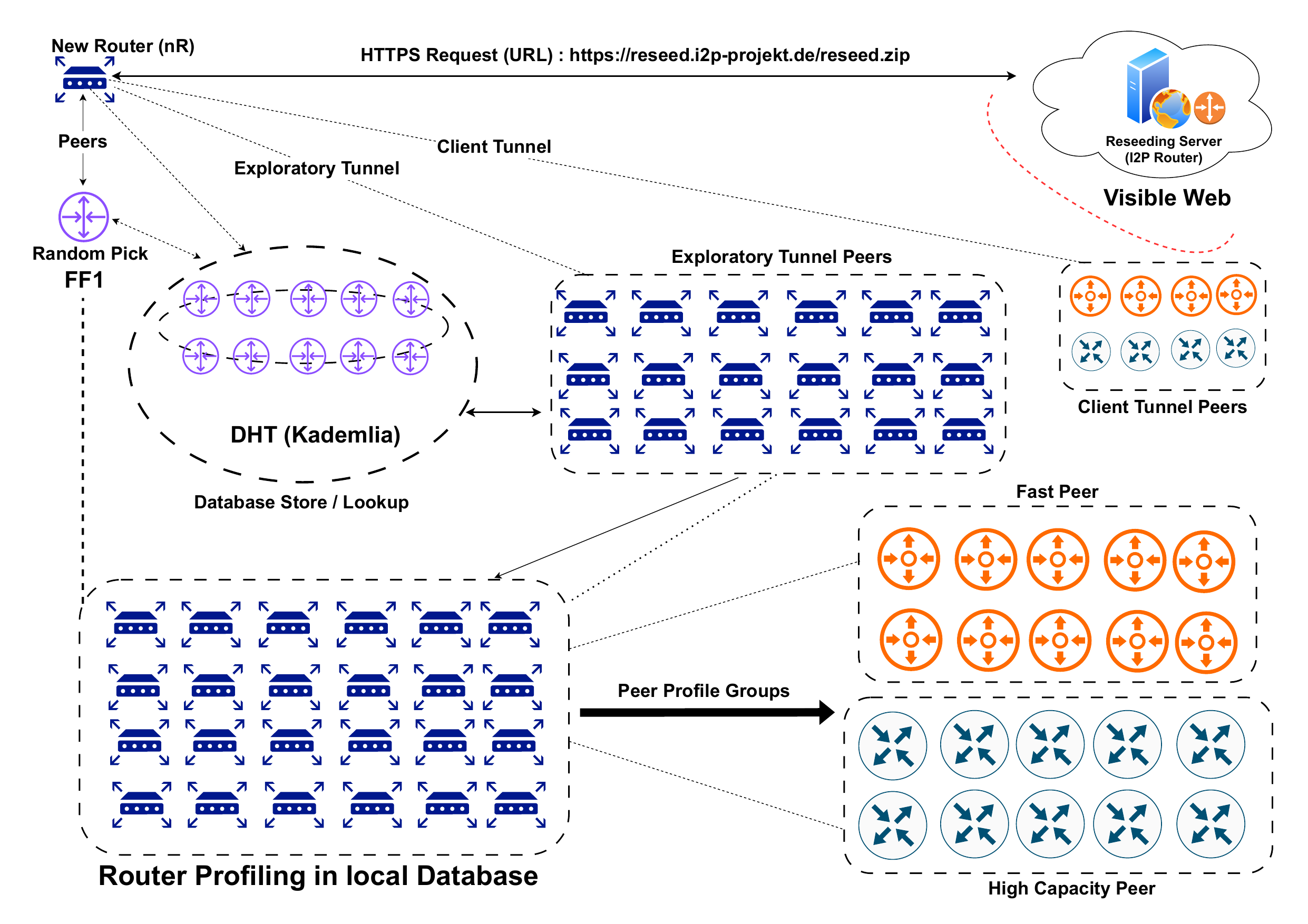}
\caption{I2P peer-selection ecosystem \cite{muntaka2026geographic}. A new router bootstraps via the public reseed mechanism, accumulates peer profiles
  in its local DHT-backed database, and selects tunnel hops
  from the fast and high-capacity peer groups produced by
  the ProfileOrganizer. \sys{} replicates this full cycle
  in an isolated environment without the external reseed
  dependency.}
\label{fig:i2p-arch}
\end{figure}

Prior I2P studies have examined netDb crawling~\cite{timpanaro11},
practical attacks including Sybil and intersection
variants~\cite{egger13}, peer-profile bias~\cite{schimmer09,herrmann11},
resilience properties~\cite{muntaka2025itl}, and the construction of
validated data-collection frameworks~\cite{muntaka2025dib}. Each of
these required bespoke testbed setups that are not reusable.
Jansen and Hopper's Shadow~\cite{shadow} provides excellent simulation
for Tor but has no I2P support. Ahmed et al. ~\cite{ahmed2025open} package a 5G core as a reproducible cyber range, establishing that the reproducibility ethic is achievable for complex network systems. No equivalent exists for I2P on the unmodified reference router in a fully isolated, ethics-safe environment.

\section{The \sys{} Emulator}
\label{sec:design}

\sys{} runs a complete I2P network inside a single machine. It stands up
dozens of real routers, wires them into a private overlay that mirrors the
live network's structure, and gives the operator full observational control
over every router, tunnel, and peer.

Three properties drove its design. \textbf{Fidelity}: it runs the unmodified
Java I2P reference implementation, not a model of it, so peer selection,
tunnel construction, NetDB behaviour, and transport handshakes are the
protocol's own. \textbf{Isolation}: the network is closed by construction,
which serves ethics and repeatability at once, since it can neither leak into
the live network nor be perturbed by it, and an experiment reproduces months
later. \textbf{Deployability}: it installs from one Debian package with source
available \cite{siddiquei2p-emulator}, launches through a graphical interface,
and rebuilds reproducibly from a single topology file, without the operator
needing to know namespaces or I2P internals.

\subsection{Architecture and Per-Router Isolation}
The reference deployment runs on the University of Cincinnati Ohio Cyber
Range Institute (OCRI), inside a nested virtualization stack
(Fig.~\ref{fig:arch}) whose innermost Ubuntu~24.04 guest hosts the emulator
with 48\,GB RAM, twelve virtual CPUs, 200\,GB disk, and a bridged adapter. The
stack imposes no functional restriction, since namespaces need only kernel
network-stack access; it fixes only an upper memory bound.

\begin{figure}[H]
\centering
\begin{tikzpicture}[font=\scriptsize,
  box/.style={draw,rounded corners=3pt,minimum width=6.5cm,
              minimum height=0.55cm,align=center},
  ns/.style={draw=gibNavy,rounded corners=2pt,fill=gibRed!20,
             minimum width=0.84cm,minimum height=0.44cm,
             align=center,font=\tiny}]
\node[box,fill=gibNavy!88,text=white] (e)
      {Physical server $\rightarrow$ VMware ESXi (Type-1)};
\node[box,fill=gibBlue!65,text=white,below=1.4mm of e] (w)
      {Windows Server 2019 VM};
\node[box,fill=gibBlue!30,below=1.4mm of w] (v)
      {Oracle VirtualBox (Type-2)};
\node[box,fill=gibGray,below=1.4mm of v,minimum height=1.72cm] (u){};
\node[anchor=north west,xshift=4pt,yshift=-3pt,font=\scriptsize]
      at (u.north west)
      {\textbf{Ubuntu 24.04 guest}\quad 48\,GB RAM, 12\,vCPU, 200\,GB};
\node[ns] (n1) at ($(u.center)+(-2.65,-0.17)$) {ns\\$r_1$};
\node[ns,right=1.1mm of n1] (n2) {ns\\$r_2$};
\node[ns,right=1.1mm of n2] (n3) {ns\\$r_3$};
\node[right=1.2mm of n3,font=\tiny] (el) {$\cdots$};
\node[ns,right=1.2mm of el] (n50){ns\\$r_{50}$};
\node[font=\tiny,below=0.7mm of n2,xshift=8mm,text=gibNavy]
      {unmodified Java I2P $\times 50$};
\begin{scope}[on background layer]
\node[draw=gibRed,dashed,rounded corners,fill=gibRed!5,
      fit=(n1)(n50),inner sep=3pt]{};
\end{scope}
\end{tikzpicture}
\caption{Deployment stack. Fifty unmodified Java I2P routers run as
isolated network-namespace processes inside a single Ubuntu guest.}
\label{fig:arch}
\end{figure}
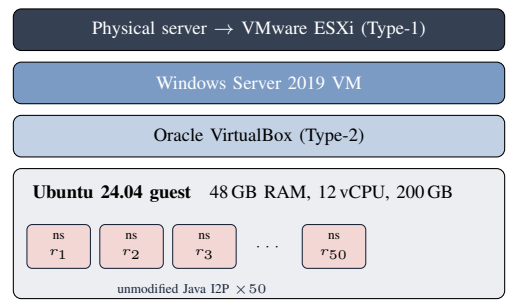

Each router $r_i$ occupies its own Linux network namespace, \code{i2pns-r$i$}, joined to an emulated subnet by a \code{veth} pair and a bridge, letting many
real routers coexist on one host without interference. Each carries its own routing table, a stable IP from the topology, and its own NTCP2 and SSU2 ports.
Inter-router packets traverse the host loopback at memory-copy speed and never reach a physical interface. Per-router state lives under a predictable tree,
\code{\textasciitilde/i2p-testnet-50/r$i$/\{config,data,logs\}}, read directly by the measurement pipeline. Isolation in I2P matters not at the binary layer
but at the identity, NetDB, tunnel, and peer-state layer; separating all of it per router lets one installation behave as fifty independent participants.

\subsection{Scale and Geographic Realism}
The reference network runs fifty routers across five locations, France,
Germany, Japan, Lebanon, and the United States, with five floodfills forming the NetDB backbone. The geography is functional: I2P
weighs country of origin when building tunnels and ranks routers into bandwidth
tiers, so \sys{} assigns each a public-style IP from a distinct address family
and a country-appropriate bandwidth profile, engaging that same selection
logic. Fifty routers is enough for it to bite, giving floodfill selection real
diversity, making tunnel paths non-trivial, and letting churn drop several
routers at once. 
Scale is bounded by the memory available to the guest, not by the platform, so a larger allocation raises the count; Section~\ref{sec:eval} reports behaviour as that bound is approached.

\subsection{Bootstrapping Without Signed Reseed}
Standard I2P bootstraps from cryptographically signed \code{.su3} bundles served
by registered reseed hosts (the reseed path in
Fig.~\ref{fig:i2p-arch})~\cite{i2p-techdoc, muntaka2026geographic}. A private
network cannot produce those bundles, and a router below its known-peers
threshold loops on the public reseed servers, the very hosts an isolated network
must never touch. \sys{} resolves this with a two-phase filesystem bootstrap that
replaces the network reseed with a local one. Phase~A, at deploy time, writes
each router's identity and publishes its RouterInfo into every peer's netDb
directory, so each router starts past the minimum-known-peers threshold and never
seeks an external reseed. Phase~B, at runtime, runs a timer that collects the
fresh RouterInfo files routers generate and redistributes them across the fleet,
which I2P's directory watcher ingests live without a restart, keeping a current
shared peer view from the local filesystem alone. The network reaches full
emergence with no external contact. Algorithm~\ref{alg:deploy} gives the full
sequence.

\begin{algorithm}[t]
\caption{Isolated \sys{} deployment}
\label{alg:deploy}
\small
\begin{algorithmic}[1]
\Require Topology $T=(R,F,G)$: routers $R$, floodfills $F\subseteq R$,
  region map $G$
\State create bridges and subnets for all blocks in $T$
\For{each router $r_i \in R$}
  \State create namespace \code{i2pns-r$i$}, veth pair, assign IP
  \State generate identity keypair; write \code{router.config}
  \State set floodfill flag if $r_i\in F$; set bandwidth from $G(r_i)$
\EndFor
\For{each router $r_i \in R$}\Comment{Phase A}
  \State copy $\{\text{RouterInfo}_j: j\neq i\}$ into
    \code{r$i$/data/netDb}
\EndFor
\State start all routers; arm cross-pollination timer (Phase B)
\State \textbf{wait} $\approx 25$\,min, undisturbed
\end{algorithmic}
\end{algorithm}

\subsection{Containment and Isolation Guarantee}
Isolation is structural, enforced by three independent mechanisms. The local
bootstrap removes the only reason a healthy router would contact an outside host;
each router's namespace has no default route to a physical interface, only the
internal bridge, so no path off the machine exists; and all addressing comes from
the topology on virtual interfaces, so packets are copied through host memory and
never reach hardware. The virtualization layers add separation, but the guarantee
holds at the namespace layer without them. We confirmed this empirically: with
the full network building tunnels, a capture on the guest's physical adapter
showed no I2P traffic to any external address, while a capture on the internal
bridge showed the expected NTCP2 and SSU2 exchange. No public reseed host,
floodfill, or off-topology address appeared, and the network sustained hours of
measurement without a single packet crossing outward. \sys{} is therefore a true,
high-fidelity I2P, real software producing real behaviour, entirely within a box.

\sys{} is run by a PyQt5 application (Fig.~\ref{fig:fleet}). The Fleet screen provides router cards that show peer and tunnel counts, reachability, uptime and console access. In the Topology Builder (Fig.~\ref{fig:topobuild}), locations, subnets, and router and floodfill counts are validated to create the sole topology JSON that drives deployment. The Topology screen (Fig.~\ref{fig:topoview}) provides the live graph visualization with the floodfills shown as well. The Network Map screen (Fig.~\ref{fig:netmap}) maps the routers geographically with live overlays and floodfill backbone. The Scenarios screen (Fig.~\ref{fig:scenarios}) runs churn, attack and resilience templates, and the Measurements screen starts and exports fleet-wide measurements. The header polls aggregated counters every eight seconds. As all the state is read from the router directory trees, the measurements are based on the router state, not inference, providing a ground truth.

\begin{figure}
\centering
\includegraphics[width=0.48\textwidth]{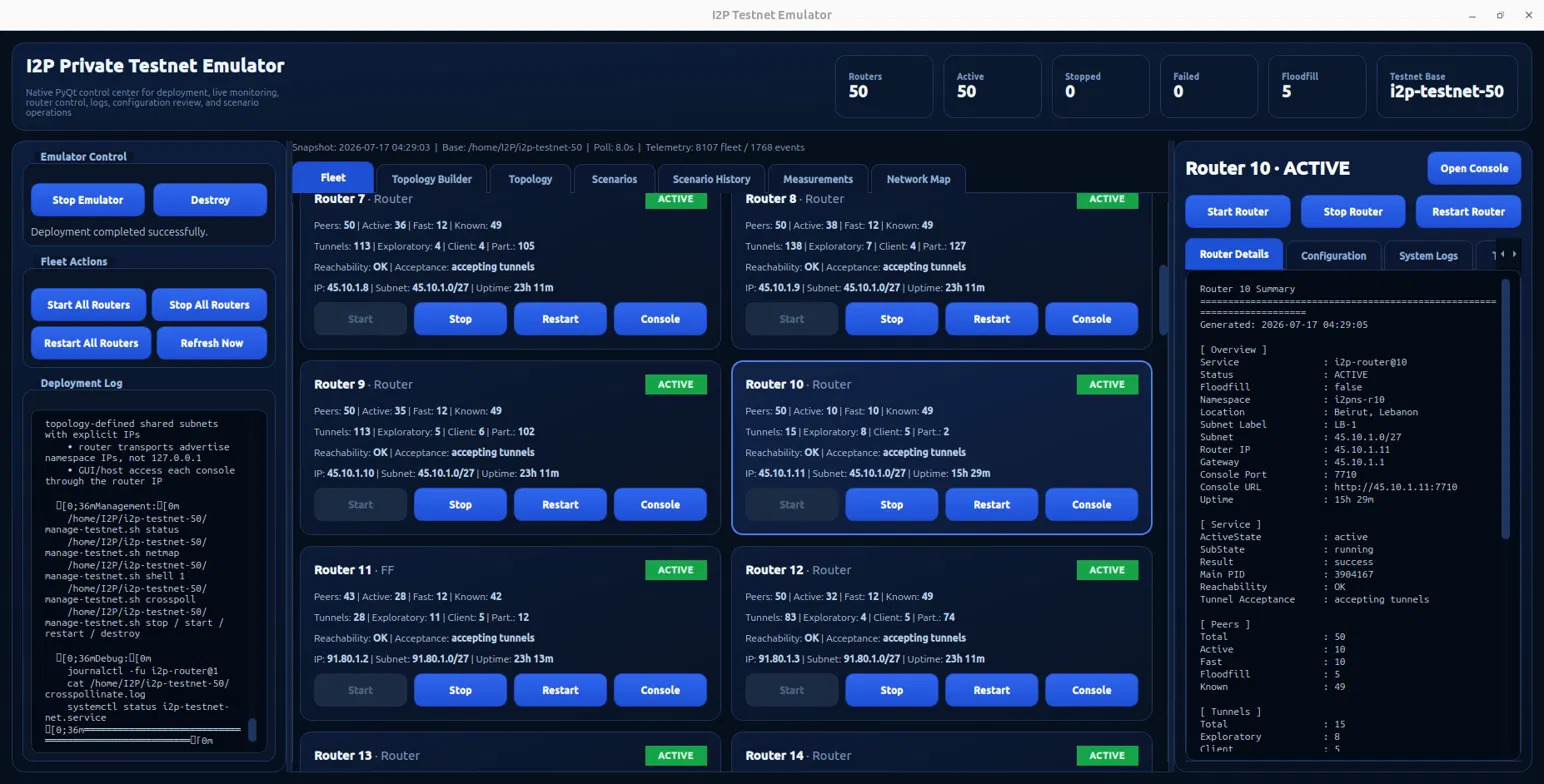}
\caption{\sys{} fleet dashboard. All 50 routers are ACTIVE. Router~10
  (the victim, highlighted) is selected; its detail pane confirms
  Floodfill:~false, Location:~Beirut,~Lebanon, IP:~45.10.1.11, and
  49 known peers.}
\label{fig:fleet}
\end{figure}

\begin{figure}
\centering
\includegraphics[width=0.95\linewidth]{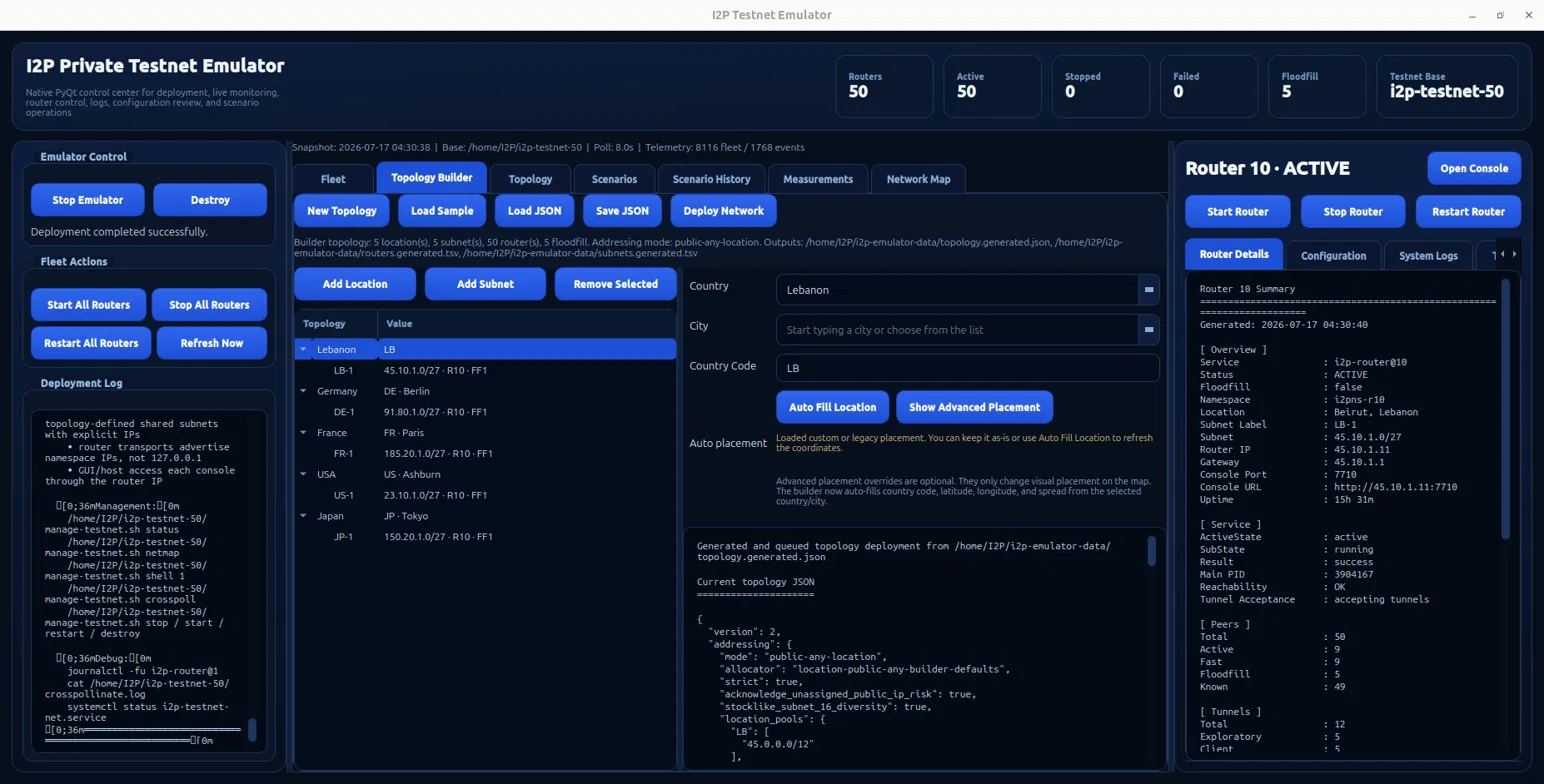}
\caption{Topology Builder. The table specifies five locations, five
  subnets, 50 routers, and five floodfills. The generated topology
  JSON (shown at bottom) is the single input to deployment.}
\label{fig:topobuild}
\end{figure}

\begin{figure}
\centering
\includegraphics[width=0.95\linewidth]{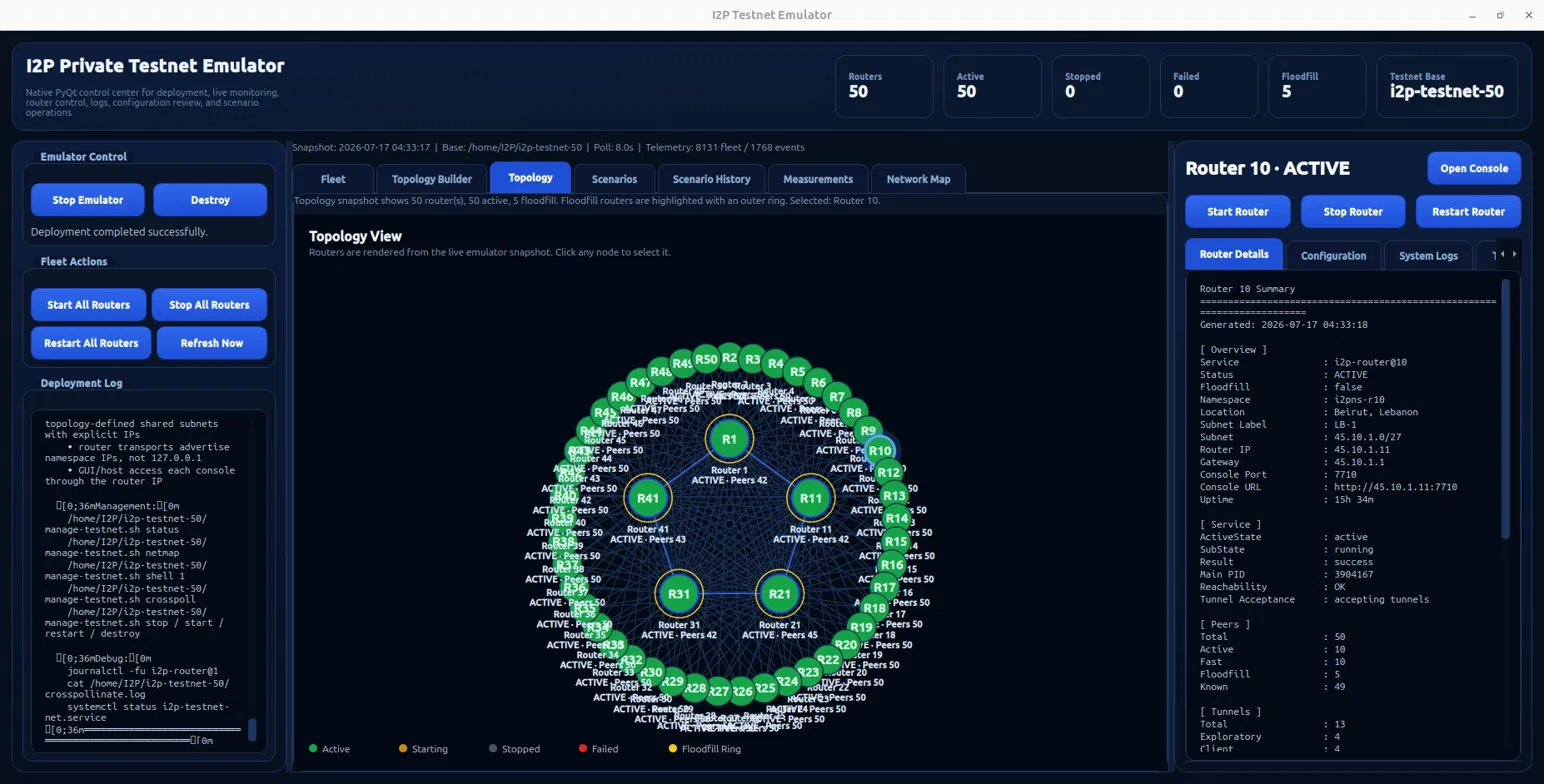}
\caption{Live topology graph of the fifty-router network after full
  emergence. All 50 nodes are active (green). The five floodfill routers
  ($r_1, r_{11}, r_{21}, r_{31}, r_{41}$) appear with a bright outer ring.
  Every node shows live peer counts.}
\label{fig:topoview}
\end{figure}

\begin{figure}
\centering
\includegraphics[width=0.95\linewidth]{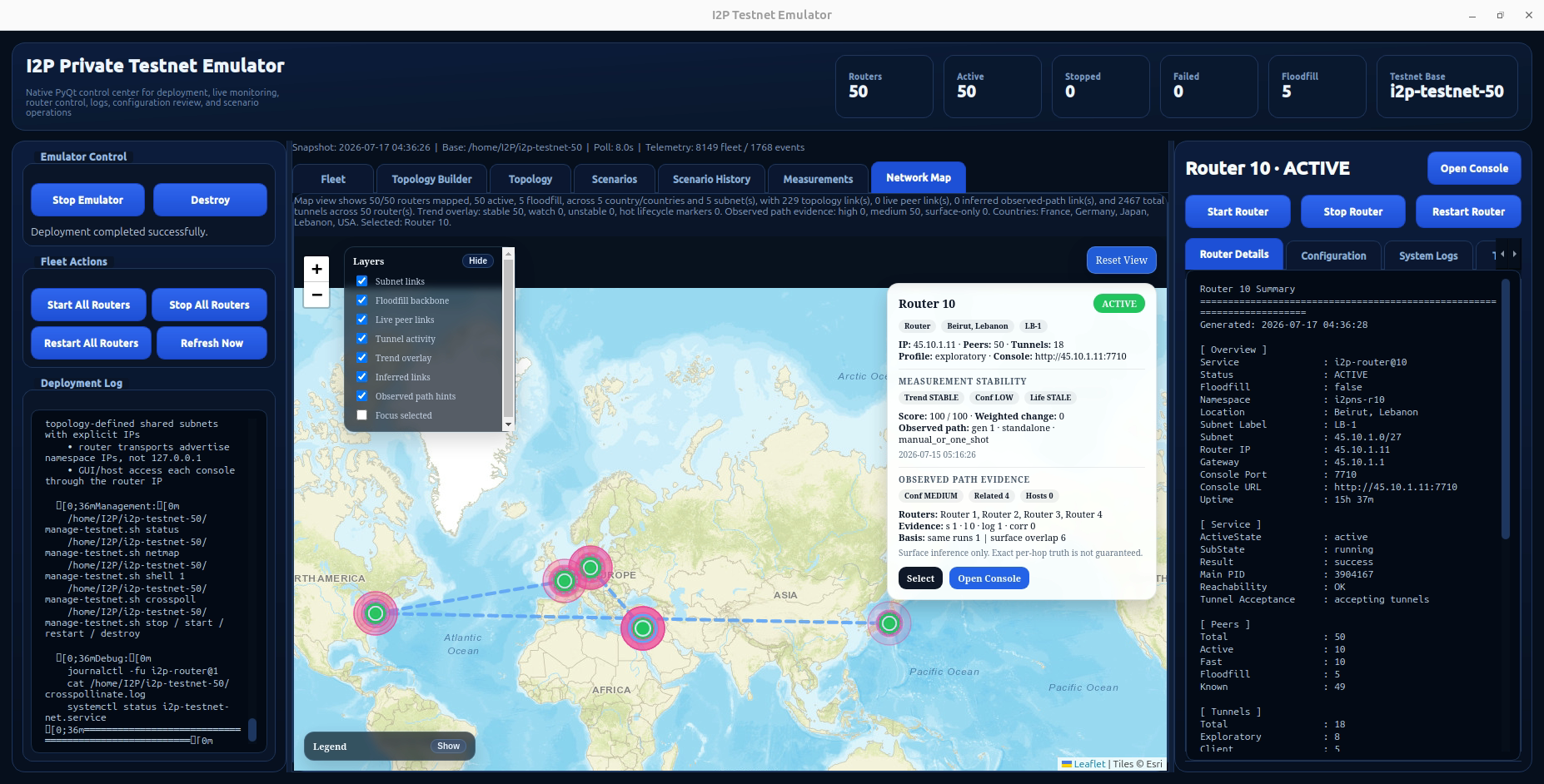}
\caption{Network Map view placing all 50 routers on a world map across
  five geographic locations (France, Germany, Japan, Lebanon, USA). The
  floodfill backbone (dashed lines) connects the five floodfill routers.
  Router~10, selected here, shows live measurement-stability data and
  observed-path evidence in its popup panel.}
\label{fig:netmap}
\end{figure}

\begin{figure}
\centering
\includegraphics[width=0.95\linewidth]{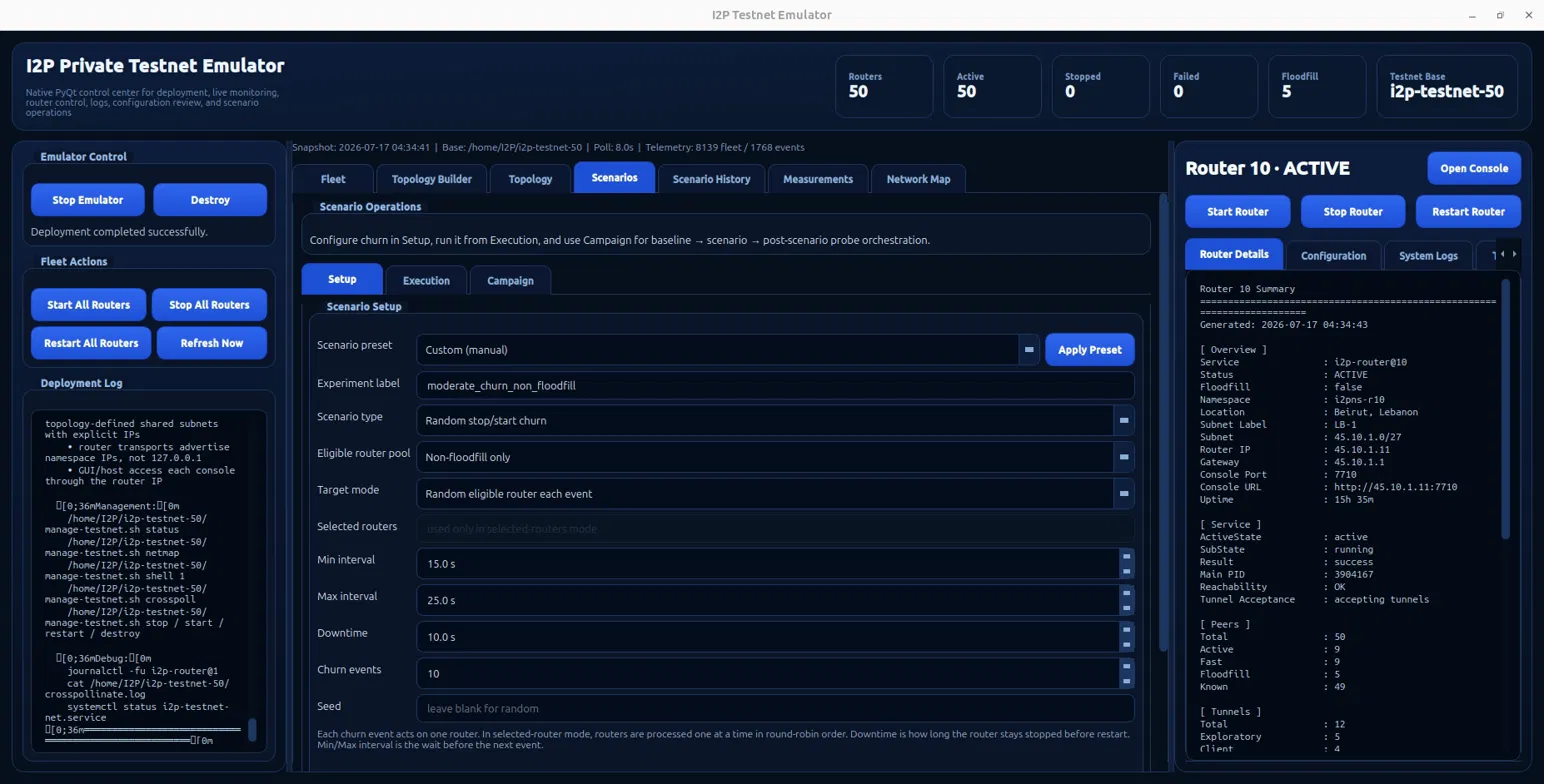}
\caption{Scenario configuration panel for a churn experiment: random
  stop/start events on non-floodfill routers, with adjustable intervals
  (15--25\,s), downtime (10\,s), and event count (10). The Setup,
  Execution, and Campaign subtabs structure the full experiment lifecycle.}
\label{fig:scenarios}
\end{figure}

\section{Deployment}
\label{sec:deploy}
The study deploys fifty routers across five locations; France, Germany, Japan, Lebanon, and the United States. It uses ten routers per location, each location occupying its own /27 subnet with one floodfill. A generated GeoIP file binds each subnet to a country code, and every router carries a matching bandwidth profile, so I2P treats the network as geographically diverse and applies its hop-per-address-block cap~\cite{muntaka2026geographic} to every tunnel built. This layout is a study configuration rather than a platform limit: the Topology Builder (Fig.~\ref{fig:topobuild}) generates arbitrary locations, subnets, and counts with a matching GeoIP mapping. Fig.~\ref{fig:topo} shows the result: five subnet clusters joined by bidirectional namespace
interconnects, each cluster's floodfill replicating the network's netDb distribution. The victim $r_{10}$ and attackers $r_{45}$--$r_{49}$ sit in separate clusters, so every tunnel the attack touches crosses subnet boundaries and exercises the address-diversity constraint.

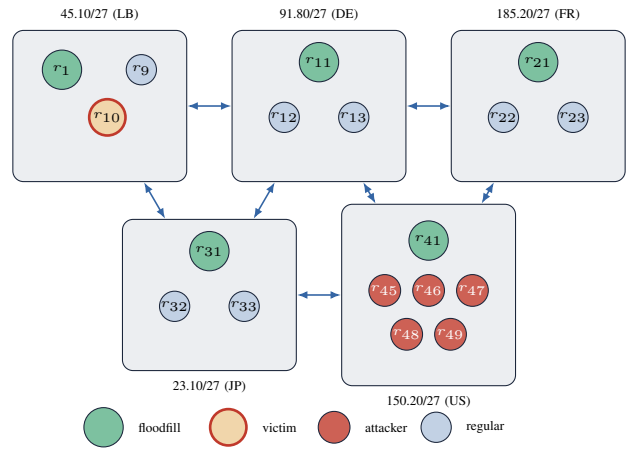
\begin{figure}[H]
\centering
\begin{tikzpicture}[font=\scriptsize,
  cl/.style={draw=gibNavy,rounded corners=4pt,minimum width=2.3cm,
             minimum height=2.0cm,fill=gibGray},
  cltall/.style={draw=gibNavy,rounded corners=4pt,minimum width=2.3cm,
             minimum height=2.4cm,fill=gibGray},
  ff/.style={circle,draw=gibNavy,fill=gibGreen!60,minimum size=0.52cm,
             font=\tiny\bfseries,inner sep=0pt},
  rr/.style={circle,draw=gibNavy,fill=gibBlue!30,minimum size=0.40cm,
             font=\tiny,inner sep=0pt},
  att/.style={circle,draw=gibNavy,fill=gibRed!80,minimum size=0.42cm,
              font=\tiny,text=white,inner sep=0pt},
  vic/.style={circle,draw=gibRed,line width=1pt,fill=gibAmber!35,
              minimum size=0.48cm,font=\tiny\bfseries,inner sep=0pt}]
\node[cl,label={[font=\tiny]above:45.10/27 (LB)}] (A) at(-2.9,1.2){};
\node[ff]  at($(A)+(-0.50, 0.48)$){$r_1$};
\node[vic] at($(A)+(0.10,-0.15)$){$r_{10}$};
\node[rr]  at($(A)+(0.55, 0.48)$){$r_9$};
\node[cl,label={[font=\tiny]above:91.80/27 (DE)}] (B) at(0,1.2){};
\node[ff]  at($(B)+(0.00, 0.58)$){$r_{11}$};
\node[rr]  at($(B)+(-0.46,-0.15)$){$r_{12}$};
\node[rr]  at($(B)+(0.46,-0.15)$){$r_{13}$};
\node[cl,label={[font=\tiny]above:185.20/27 (FR)}] (C) at(2.9,1.2){};
\node[ff]  at($(C)+(0.00, 0.58)$){$r_{21}$};
\node[rr]  at($(C)+(-0.46,-0.15)$){$r_{22}$};
\node[rr]  at($(C)+(0.46,-0.15)$){$r_{23}$};
\node[cl,label={[font=\tiny]below:23.10/27 (JP)}] (D) at(-1.45,-1.3){};
\node[ff]  at($(D)+(0.00, 0.58)$){$r_{31}$};
\node[rr]  at($(D)+(-0.46,-0.15)$){$r_{32}$};
\node[rr]  at($(D)+(0.46,-0.15)$){$r_{33}$};
\node[cltall,label={[font=\tiny]below:150.20/27 (US)}] (E) at(1.45,-1.3){};
\node[ff]  at($(E)+(0.00, 0.72)$){$r_{41}$};
\node[att] at($(E)+(-0.58, 0.06)$){$r_{45}$};
\node[att] at($(E)+(0.00, 0.06)$){$r_{46}$};
\node[att] at($(E)+(0.58, 0.06)$){$r_{47}$};
\node[att] at($(E)+(-0.29,-0.52)$){$r_{48}$};
\node[att] at($(E)+(0.29,-0.52)$){$r_{49}$};
\foreach \x/\y in {A/B,A/D,B/C,B/D,B/E,C/E,D/E}
  \draw[gibBlue,semithick,{Latex[length=1.4mm]}-{Latex[length=1.4mm]}]
    (\x)--(\y);
\node[ff] (l1) at(-2.85,-3.05){};\node[right=2pt of l1,font=\tiny]{floodfill};
\node[vic](l2) at(-1.20,-3.05){};\node[right=2pt of l2,font=\tiny]{victim};
\node[att](l3) at(0.20,-3.05){};  \node[right=2pt of l3,font=\tiny]{attacker};
\node[rr] (l4) at(1.55,-3.05){}; \node[right=2pt of l4,font=\tiny]{regular};
\end{tikzpicture}
\caption{Logical topology of the fifty-router deployment. Five subnet
clusters span five geographic locations. The five attackers
($r_{45}$--$r_{49}$) share the US cluster; the victim ($r_{10}$) sits in a
separate cluster.}
\label{fig:topo}
\end{figure}

\section{Testing}
\label{sec:testing}

\subsection{A Two-Stage Model of Peer Selection}
\label{sec:model}
Tunnel participation unfolds in two stages. In \emph{selection}, the victim
decides which peers are eligible to appear in its tunnels; in \emph{delivery},
an eligible peer must accept the build request and actually relay traffic.
Keeping the stages apart is what lets the results be read precisely, since the
attack we study succeeds at selection and then exposes what delivery
additionally demands.

\begin{definition}[Peer profile]
Let $v$ be a victim router and $\mathcal{P}_v$ its known peers. For each
$p\in\mathcal{P}_v$, the profile is a tuple
$\pi_v(p)=(\hat{s},\hat{c},\beta^s,\beta^c,\mathbf{h})$, where
$\hat{s},\hat{c}$ are measured speed and capacity, $\beta^s,\beta^c$ are
stored bonuses, and $\mathbf{h}$ is the tunnel history.
\end{definition}

The ProfileOrganizer adds each stored bonus to the measured estimate,
\begin{align}
s_v(p) &= \hat{s}_v(p) + \beta^s_v(p),\label{eq:s}\\
c_v(p) &= \hat{c}_v(p) + \beta^c_v(p),\label{eq:c}
\end{align}
and admits to the preferred set every peer that clears both thresholds,
\begin{equation}
\Fv = \bigl\{\,p\in\mathcal{P}_v \;\bigm|\;
  s_v(p)\ge\tau^s_v \;\wedge\; c_v(p)\ge\tau^c_v\,\bigr\}.\label{eq:fv}
\end{equation}
The thresholds $\tau^s_v,\tau^c_v$ follow high percentiles of the score
distributions the victim observes.

\begin{definition}[Selection probability]
For a tunnel of hop length $\ell$, a peer $p\in\Fv$ appears in any given
tunnel of $v$ with probability
\begin{equation}
\Pr[p\in T_v]\approx\frac{\ell}{|\Fv|},
\label{eq:psel}
\end{equation}
subject to address-diversity constraints, and zero if $p\notin\Fv$.
\end{definition}

Entering $\Fv$ is the adversary's goal at selection, and
Eqs.~\eqref{eq:s} and~\eqref{eq:c} expose two independent routes. The
\textit{behavioral lever} raises the measured terms $\hat{s},\hat{c}$ through
genuine performance: advertising high bandwidth, accepting builds quickly,
holding uptime. The \textit{injection lever} writes large bonuses
$\beta^s,\beta^c$ directly, skipping measurement altogether.

\begin{theorem}[Injection guarantees eligibility]
If $\beta^s_v(a)>\tau^s_v$ and $\beta^c_v(a)>\tau^c_v$, then
$a\in\Fv$ for any $\hat{s}_v(a),\hat{c}_v(a)\ge 0$.
\end{theorem}
\begin{proof}
Substituting into~\eqref{eq:fv}: $s_v(a)=\hat{s}_v(a)+\beta^s_v(a)>
\tau^s_v$ and $c_v(a)=\hat{c}_v(a)+\beta^c_v(a)>\tau^c_v$, hence
$a\in\Fv$.
\end{proof}

Eligibility is not participation. A peer in $\Fv$ must still accept the build
and carry traffic. Writing $\alpha_v(p)$ for the rate at which $p$ agrees to
join tunnels proposed by $v$, the expected volume through $p$ is
\begin{equation}
\mathbb{E}[\text{tunnels through }p]\propto
  \Pr[p\in T_v]\cdot\alpha_v(p).\label{eq:del}
\end{equation}
An adversary that maximises $\Pr[p\in T_v]$ while $\alpha_v(p)$ stays low
gains little, and Section~\ref{sec:eval} measures both factors to show exactly
that. For each profile metric $m$ we report the attacker-to-control ratio
\begin{equation}
\rho(m)=\frac{\overline{m}_A}{\overline{m}_C},\label{eq:rho}
\end{equation}
averaged over trials, with values above one favouring the adversary.

\subsection{Measurement Methodology}
\label{sec:method}
I2P writes the very quantities of Section~\ref{sec:model} to disk. For each
known peer the victim keeps a gzipped file,
\code{r$v$/data/peerProfiles/p$X$/profile-\textit{hash}.txt.gz}, holding the
fields in Table~\ref{tab:fields}, so reading these files exposes the selection
state faithfully and without modifying the router.

\begin{table}[t]
\centering\footnotesize
\caption{Profile fields used for measurement}
\label{tab:fields}
\begin{tabular}{lp{3.2cm}}
\toprule
\textbf{Field} & \textbf{Role in selection} \\
\midrule
\code{speedBonus}               & injected bonus $\beta^s$ (stage one)\\
\code{capacityBonus}            & injected bonus $\beta^c$ (stage one)\\
\code{tunnelPeakThroughput}     & measured speed $\hat{s}$ \\
\code{tunnels.lifetimeAgreedTo} & accepted builds (delivery proxy)\\
\code{tunnels.lifetimeRejected} & rejected builds (delivery proxy)\\
\code{tunnelCreateResponseTime} & response latency\\
\bottomrule
\end{tabular}
\end{table}

A live router holds profile updates in memory and flushes them only on graceful
shutdown. To take a consistent snapshot we stop the victim, which forces the
flush, copy its \code{peerProfiles} tree and \code{netDb}, and restart it; each
flush yields 46 to 49 profiles. Profiles are keyed by I2P identity hash rather
than emulator index, so we rebuild the mapping from the victim's own netDb:
each attacker advertises a known topology IP, and scanning RouterInfo files for
that IP yields the hash from the filename. We recompute this mapping from live
state on every trial.

This protocol was learned the hard way. An early pilot compared a fresh run
against a victim profile that still held an injection from an earlier trial, and
the $9999\to9999$ baseline that resulted was an artefact rather than a
measurement. Every trial now passes an anti-contamination gate that confirms
all attacker bonuses read zero before it proceeds. Algorithm~\ref{alg:trial}
gives the per-trial procedure in full.

\begin{algorithm}[H]
\caption{One reputation-attack trial}
\label{alg:trial}
\small
\begin{algorithmic}[1]
\Require victim $v$, attackers $A$, controls $C$
\State \textbf{clean:} revert attacker configs; zero
  $\beta^s_v(a),\beta^c_v(a)$ for all $a\in A$; settle 10\,min
\State \textbf{baseline:} stop $v$; snapshot profiles; restart $v$
\State \textbf{gate:} assert $\beta^s_v(a)=0\;\forall a\in A$,
  else abort\label{ln:gate}
\State \textbf{attack:} apply behavioral config to $A$; stop $v$;
  write $\beta^s,\beta^c\leftarrow 9999$; start $v$; settle 10\,min
\State \textbf{attack snapshot:} stop $v$; copy profiles; restart $v$
\State \textbf{analyze:} resolve hashes; compute
  $\overline{m}_A,\overline{m}_C,\rho(m)$ per Table~\ref{tab:fields}
\end{algorithmic}
\end{algorithm}

Each trial runs end to end without pause, because a pause lets the coalescing
timer overwrite the injected values; we observed this in early work and treat
it as the main reliability hazard.


\subsection{Evaluation}
\label{sec:eval}
The evaluation spans five independent trials. Router $r_{10}$
(IP \code{45.10.1.11}) is the victim throughout. Five regular non-floodfill
routers in a separate subnet cluster, $r_{45}$ through $r_{49}$, act as
attackers, and five matched regular routers, $r_{42}$ through $r_{44}$ with
$r_{32}$ and $r_{33}$, serve as controls. The anti-contamination gate
(Algorithm~\ref{alg:trial}, line~\ref{ln:gate}) passed in all five trials, so
every baseline began with attacker bonuses at zero.
Table~\ref{tab:res} reports the footprint from live systemd cgroup readings.
Memory is the binding resource. At 417\,MB per router the fleet settles near
21\,GB, leaving the 48\,GB guest ample headroom, enough in principle for nearly
ninety routers. We stopped at fifty because five geographic locations already
made the topology representative, not because of that ceiling.

\begin{table}[t]
\centering\footnotesize
\caption{Resource footprint of the fifty-router deployment}
\label{tab:res}
\begin{tabular}{lrr}
\toprule
\textbf{Resource} & \textbf{Per router} & \textbf{Fleet (50)} \\
\midrule
RAM (measured)      & $\approx 417$\,MB & $\approx 20.9$\,GB\\
Disk (testnet tree) & $\approx 9$\,MB   & $\approx 457$\,MB\\
Emergence time      & ---               & $\approx 25$\,min\\
Single-host ceiling & ---               & $\approx 90$ routers\\
\bottomrule
\end{tabular}
\end{table}

Across all five trials the injected values held exactly. Attacker
\code{speedBonus} and \code{capacityBonus} read precisely $9999$ in the
victim's profiles with no cross-run variance (Fig.~\ref{fig:inj}), while the
controls remained at zero. Table~\ref{tab:main} gives the full breakdown. This
is Theorem~1 in practice: once a bonus clears the threshold, membership in
$\mathcal{F}_v$ follows on the unmodified router every time.

\begin{figure}[H]
\centering
\includegraphics[width=0.80\linewidth]{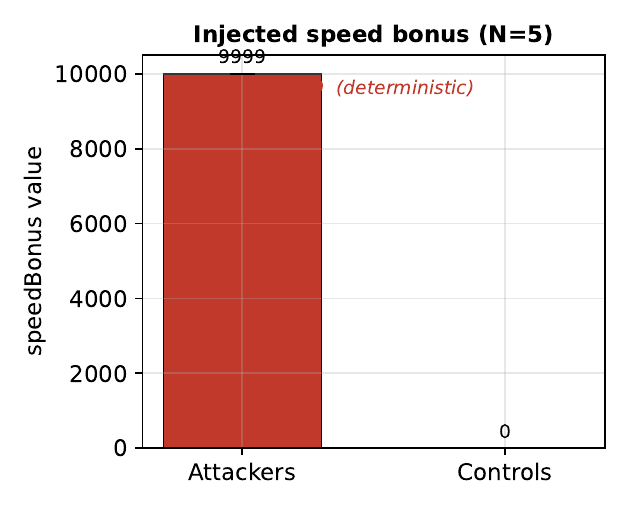}
\caption{Injected \code{speedBonus} in the victim's profiles, five
  trials. The effect is perfectly repeatable ($\sigma=0$); controls
  are untouched.}
\label{fig:inj}
\end{figure}

\begin{table}[t]
\centering\footnotesize
\caption{Mean $\pm$ std across five trials (attack condition).
  Baseline bonuses are zero by the gate.}
\label{tab:main}
\begin{tabular}{lrr}
\toprule
\textbf{Metric} & \textbf{Attackers} & \textbf{Controls} \\
\midrule
\code{speedBonus}           & $9999.0\pm0.0$ & $0.0\pm0.0$  \\
\code{capacityBonus}        & $9999.0\pm0.0$ & $0.0\pm0.0$  \\
\code{tunnelPeakThroughput} & $0.127\pm0.000$& $0.210\pm0.072$\\
\code{lifetimeAgreedTo}     & $10.76\pm0.20$ & $7.80\pm1.84$\\
\code{lifetimeRejected}     & $265.8\pm21.2$ & $77.0\pm16.5$\\
\code{responseMs}           & $21.71\pm4.82$ & $16.52\pm1.29$\\
\bottomrule
\end{tabular}
\end{table}

Eligibility carried through into selection outcomes. The attackers drew
$10.76\pm0.20$ lifetime build requests on average against the controls'
$7.80\pm1.84$, a ratio $\rho(\text{agreedTo})=1.38$ (Fig.~\ref{fig:sel}, left),
matching what Eq.~\eqref{eq:psel} predicts: the victim invites the adversary
into its tunnels at $1.38\times$ the rate of comparable peers.

\begin{figure}[t]
\centering
\includegraphics[width=0.99\linewidth]{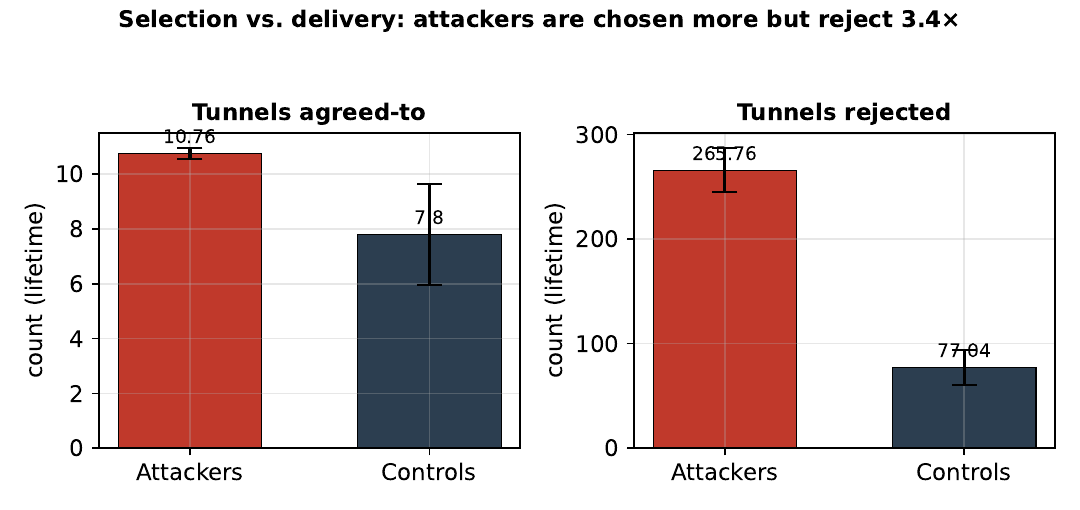}
\caption{Delivery stage. Attackers receive more build requests (left)
  but reject a larger fraction (right), illustrating the separation
  between the two stages.}
\label{fig:sel}
\end{figure}

Delivery behaved differently, and the two-stage separation shows here plainly.
Although selected more often, the attackers accepted far fewer requests: their
lifetime rejections reached $265.8\pm21.2$ against the controls' $77.0\pm16.5$,
giving $\rho(\text{rejected})=3.45$ (Fig.~\ref{fig:sel}, right).
Equation~\eqref{eq:del} explains why. The aggressive behavioral configuration
we applied, maximising advertised bandwidth and the participating-tunnel limit,
pushed the attackers into a regime where they declined most builds. Selection
succeeded completely; delivery demanded conditions the configuration did not
meet. Compromising the selection criterion is therefore a necessary step toward
traffic interception, not a sufficient one.

The behavioral lever proved weak under light load. Attacker throughput measured
$0.127\pm0.000$\,KBps, below the controls' $0.210\pm0.072$\,KBps, for
$\rho(\text{throughput})=0.60$ (Fig.~\ref{fig:thru}). Advertising a 16\,MB/s
bandwidth class did not lift the victim's measurement, because measured
throughput reflects traffic actually carried and testbed traffic is light: the
lever moved $\beta$ but not $\hat{s}$. Injection is thus the effective component
here, and whether the behavioral lever behaves differently on the loaded network is exactly what the emulator is suited to study.

\begin{figure}[t]
\centering
\includegraphics[width=0.80\linewidth]{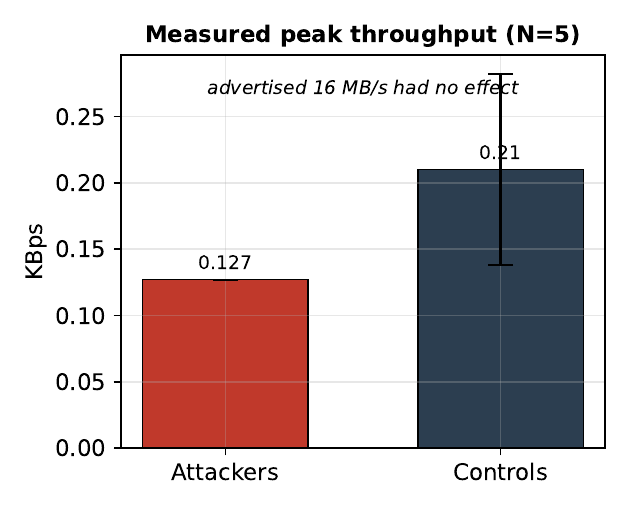}
\caption{Measured peak throughput, five trials. Advertising a high
  bandwidth class did not raise the attackers' measured throughput above
  the controls.}
\label{fig:thru}
\end{figure}

Fig.~\ref{fig:rat} collects $\rho(m)$ across all metrics. Selection moves
strongly for the adversary while delivery does not follow of its own accord,
and together the trials locate the structure of the attack with the precision
the platform was built to provide.

\begin{figure}[t]
\centering
\includegraphics[width=0.80\linewidth]{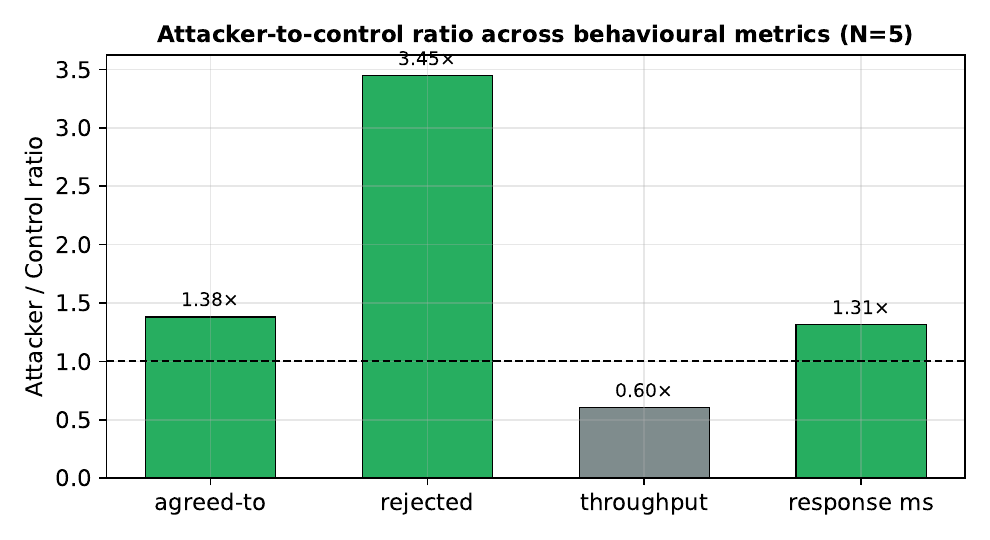}
\caption{Attacker-to-control ratios $\rho(m)$, five trials. Values
  above the dashed line ($\rho=1$) favour the adversary.}
\label{fig:rat}
\end{figure}

\section{Cybersecurity and Cyberwarfare Training}
\label{sec:training}

A CISA advisory issued in April 2026 characterised state-sponsored covert
relay ecosystems as a strategic national security priority~\cite{cisa2026}.
ORB variants used by groups such as APT28 and Volt Typhoon adapt
peer-selection and relay-management techniques drawn from conventional
anonymity networks including I2P~\cite{muntaka2025dib}. Training cyber
operators to recognise and disrupt these architectures therefore requires
hands-on exposure to the real protocol mechanics, not classroom
abstractions. \sys{} provides this substrate.

\subsection{Attack Scenario}

The Scenarios panel provides configurable attack and resilience templates, each specifying an attacker set, scenario type, timing parameters, and event count. Templates
can run from the interface or headlessly from a script. The current
library includes: \textbf{reputation injection} (studied in this paper);
\textbf{Sybil flooding}, which bursts attacker RouterInfos to study netDb
poisoning and floodfill recovery; \textbf{bandwidth inflation}, which
tests whether advertised capacity alone wins the high-capacity tier; and
\textbf{churn}, which stops or degrades routers to measure tunnel-rebuild
latency and anonymity-set reduction under attrition.

\subsection{Red-Team and Blue-Team Exercises}
The platform supports structured two-team exercises
(Fig.~\ref{fig:training}). A red team launches a scenario. A blue team
monitors fleet telemetry, inspects profiles and logs, and tries to locate
the anomaly. A debrief follows, in which the team proposes a defence and
retests it against the same scenario. Because a fresh testnet emerges in
under thirty minutes, the cycle runs multiple times in a single session.

\begin{figure}[t]
\centering
\begin{tikzpicture}[font=\scriptsize,node distance=5mm,
  st/.style={draw=gibNavy,rounded corners,fill=gibBlue!18,
             minimum width=2.6cm,minimum height=0.55cm,align=center},
  de/.style={draw=gibNavy,diamond,fill=gibGreen!22,aspect=2.2,
             minimum height=0.48cm,align=center,inner sep=1pt},
  ar/.style={-{Latex[length=1.5mm]},semithick,gibNavy}]
\node[st] (dep) {Deploy network};
\node[st,below=of dep] (red) {Red team launches};
\node[st,below=of red] (blu) {Blue team monitors};
\node[de,below=5mm of blu] (det) {detected?};
\node[st,below=5mm of det] (rep) {Debrief and defend};
\node[st,below=of rep] (pat) {Apply and rerun};
\node[st,right=12mm of det] (esc) {Escalate scenario};
\draw[ar](dep)--(red);
\draw[ar](red)--(blu);
\draw[ar](blu)--(det);
\draw[ar](det)--node[right,font=\tiny]{yes}(rep);
\draw[ar](det)--node[above,font=\tiny]{no}(esc);
\draw[ar](rep)--(pat);
\draw[ar](pat.west)--+(-8mm,0)|-(dep.west);
\draw[ar](esc.north)|-(red.east);
\end{tikzpicture}
\caption{Red-team and blue-team training cycle supported by \sys{}.
  A detected anomaly leads to debrief and defence; an undetected one
  escalates the scenario and repeats.}
\label{fig:training}
\end{figure}
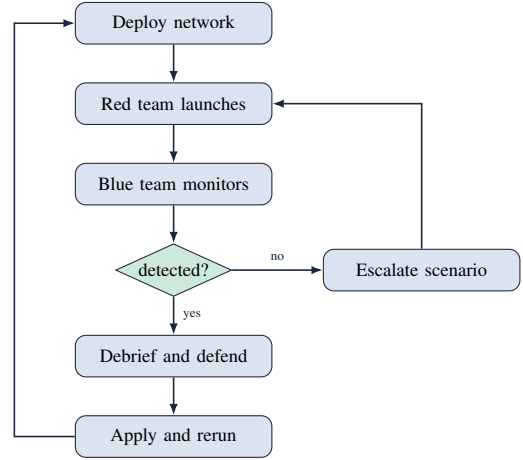

\subsection{Curriculum and Operator Readiness}
The platform supports a short teaching module. One natural progression
begins with I2P architecture and netDb mechanics explored through the live router consoles, moves to peer selection and the reputation model
with a profile-inspection exercise, then has students run the provided scenarios and interpret the output for themselves, and finally asks them to design a countermeasure and test it. Because the Scenario
History view records every run with its parameters and outcomes, a session can be revisited during after-action review.

For operator readiness the value is more specific. The platform offers a safe, protocol-authentic setting in which the peer-selection mechanics that ORB variants inherit from I2P can be observed live rather than in
abstraction. Working from real I2P router consoles, an operator can watch a reputation-poisoned peer climb in the victim's profile, see its selection probability rise as a result, and compare how its traffic signature in the logs departs from that of a legitimate peer. To the best of our knowledge, no existing tool offers this capability on an unmodified I2P router.

\section{Lessons Learned}
\label{sec:lessons}
Building and operating \sys{} surfaced several practical lessons that shaped the final design, and that others working with isolated I2P testnets are likely to encounter.

\textbf{Bootstrap is the most critical design choice.} Strip out the filesystem
seed and every router drops straight into a reseed loop. Our early prototypes
leaned on the public reseed servers, and the result was always the same: a
network in which a few routers never managed to integrate. Pairing the Phase~A
seed with the Phase~B cross-pollination timer closes off that failure entirely.

\textbf{Emergence must be uninterrupted.} If the network is paused while routers
are still building their first tunnels, peer profiles and netDb entries are left half-expired. Routers stopped and restarted mid-emergence returned with markedly lower peer counts than those left undisturbed. The twenty-five-minute
window, then, is no safety margin; it is the floor for a healthy fleet.

\textbf{Profile flush timing is an experimental variable.} A running router keeps
all of its profile data in memory, so reading the on-disk profiles while it runs
simply returns nothing. Our early experiments assumed the on-disk copies were
live and were quietly rewarded with false negatives. Stopping the router gracefully and only then snapshotting is the one approach we found reliable.

\textbf{Clean-baseline discipline is non-negotiable.} One pilot compared an
attack snapshot against a baseline that still carried injected values from an
earlier run; the $9999\to9999$ difference came out to zero and the trial reported
no effect at all. The anti-contamination gate (Algorithm~\ref{alg:trial},
line~\ref{ln:gate}) exists precisely to catch this, and it has to run on every
trial.

\textbf{Live cgroup measurement corrects static estimates.} Our first per-router memory figure, 257\,MB, came from a partial reading and was simply wrong. The live systemd cgroup measurement put it at 417\,MB, 62\% higher. The lesson generalises: measure resource use on the running system, not from static
profiling tools or back-of-envelope estimates.

\section{Discussion and Future Work}
\label{sec:disc}
Splitting the attack into two stages is what keeps the result unambiguous. An adversary that reaches $\mathcal{F}_v$ has defeated the selection criterion
outright; whether that success converts into carried traffic is a separate question with its own conditions. The practical upshot is that a defence mounted at the selection layer confronts the prerequisite shared by a broad class of interception attempts, and does so at the more tractable of the two points. The model also points to a simple defence. An injected bonus sits orders of magnitude above anything legitimate observation would produce, so a router that
checks a stored bonus against the distribution of its own measurements for that peer could flag the mismatch and discount the value, leaving peers that earned their standing untouched. Quantifying this defence, including its false-positive
cost, is a natural next step.
Several other directions follow. Instrumenting a router to emit authoritative tunnel-hop events would measure stage-two delivery directly, rather than through
the proxies we rely on now. Spreading the deployment across multiple hosts over
an encrypted mesh would push past the single-machine ceiling of roughly 90~routers. A dedicated traffic generator would exercise the behavioral lever under realistic load, where its effect may look quite different. Finally, the scenario library invites expansion: netDb eclipse and floodfill corruption are
two obvious ORB-relevant mechanics still to be covered.

\section{Conclusion}
\label{sec:concl}
We presented \sys{}, an installable emulator that runs unmodified I2P
routers at useful scale, with a control interface, a scenario engine,
and a filesystem-based measurement pipeline. The platform eliminates the
reproducibility barrier that has limited I2P security research and
provides a protocol-authentic training environment for cyber-defence
exercises on anonymous overlay threats. A five-trial evaluation of a
reputation attack showed deterministic success at the selection stage
and precise characterisation of the delivery stage. We release the
platform and its measurement scripts so others can reproduce the study,
test the defensive signal we identify, and build on the scenario library
for their own research and training.

\section*{Acknowledgment}
The authors thank the MIRAGe-UC Laboratory and the UC-OCRI (\url{http://mirage-uc.org/}) for computing resources.

\balance
\bibliographystyle{IEEEtran}
\bibliography{references}

\section*{Artifact Availability}
The emulator and measurement pipeline are released as open-source software
(MIT License) at \url{https://github.com/abksiddique/i2p-emulator}.

\section*{LLM Usage Statement}
LLMs were used for editorial purposes in this manuscript, and all outputs
were inspected by the authors to ensure accuracy and originality.

\end{document}